%% file: main.tex
\ifx\synctex\undefined\else\synctex=1\fi
\documentclass{llncs}

\PassOptionsToPackage{hyphens}{url}\usepackage[hidelinks]{hyperref}

\usepackage{ellipsis, mparhack, ragged2e} 
\usepackage[l2tabu, orthodox]{nag} 

\usepackage[english]{babel}
\usepackage[T1]{fontenc}
\usepackage{lmodern}
\usepackage[utf8]{inputenc}
\usepackage[final]{microtype}

\usepackage{cust} 

\usepackage{booktabs}

\usepackage{tabu,multirow}
\newcolumntype{L}{X}
\newcolumntype{R}{>{\raggedleft\arraybackslash}X}
\newcolumntype{C}{>{\centering\arraybackslash}X}

\usepackage{amsmath, amssymb}
\usepackage[hyperref, amsmath]{ntheorem}

\usepackage[linesnumbered]{algorithm2e}

\usepackage[labelfont=bf]{caption}
\let\llncssubparagraph\subparagraph
\let\subparagraph\paragraph
\usepackage{titlesec}
\let\subparagraph\llncssubparagraph

\usepackage{xcolor}

\usepackage{tikz}
\usetikzlibrary{automata,positioning,shapes,calc,fit}
\usepackage{pgfplots}

\usepackage[autostyle=true]{csquotes}

\usepackage[english,capitalise]{cleveref}
\crefname{theorem}{Theorem}{Theorems}\Crefname{theorem}{Theorem}{Theorems}
\crefname{lemma}{Lemma}{Lemmas}\Crefname{lemma}{Lemma}{Lemmas}
\crefname{definition}{Definition}{Definitions}\Crefname{definition}{Definition}{Definitions}
\crefname{figure}{Figure}{Figures}\Crefname{figure}{Figure}{Figures}
\crefname{remark}{Remark}{Remarks}\Crefname{remark}{Remark}{Remarks}
\crefalias{AlgoLine}{line} 
\crefformat{footnote}{#2\footnotemark[#1]#3}

\renewcommand{\baselinestretch}{0.99} 
\advance\textwidth1mm 
\advance\textheight1mm

\hypersetup{
	pdfauthor={Jan Kretinsky, Tobias Meggendorfer, Clara Waldmann and Maximilian Weininger},
	pdftitle={Index appearance record for transforming Rabin automata into parity automata},
	pdfsubject={},
	pdfkeywords={},
	colorlinks=false
}

\newcommand*{\edgestatus}[3]{\tikz[baseline=(char.base)]{\node[shape=rectangle,rounded corners=2,fill=#3,draw,text=#2,inner sep=1pt,outer sep=0pt] (char) {#1\strut};}}
\newcommand*{\edgeprohibited}[1]{\edgestatus{#1}{white}{black}}
\newcommand*{\edgerequired}[1]{\edgestatus{#1}{black}{white}}
\newcommand*{\edgepriority}[1]{\tikz[baseline=(char.base)]{\node[shape=rectangle,rounded corners=2,fill=black!10,inner sep=1pt,outer sep=0pt] (char) {#1\strut};}}

\title{Index appearance record for transforming Rabin~automata into parity automata}
\author{Jan~K{\v r}et\'insk\'y \and Tobias~Meggendorfer \and Clara~Waldmann \and Maximilian~Weininger}
\institute{Technical University of Munich}

\begin{document}

\maketitle

\vspace*{-1em}

\begin{abstract}
	Transforming deterministic $\omega$-automata into deterministic parity automata is traditionally done using variants of appearance records.
	We present a more efficient variant of this approach, tailored to Rabin automata, and several optimizations applicable to all appearance records.
	We compare the methods experimentally and find out that our method produces smaller automata than previous approaches.
	Moreover, the experiments demonstrate the potential of our method for LTL synthesis, using LTL-to-Rabin translators.
	It leads to significantly smaller parity automata when compared to state-of-the-art approaches on complex formulae.
\end{abstract}

\input{intro}
\input{defs}
\input{iar}
\input{opt}
\input{exper}

\section{Conclusion}

We have presented a new version of index appearance record.
In comparison to the standard Streett-based approach, our new Rabin-based approach produces significantly smaller automata.
Besides, it has a significant potential for LTL synthesis.
For more complex formulae, it makes use of high efficiency of Rabinizer and thus avoids the blow-up in many cases, compared to determinization-based methods.

Since we only provided the method for DRA we want to further investigate whether it can be extended to DGRA more efficiently than by de-generalization.
Besides, a more targeted post-processing of the state space and the priority function is desirable.
For instance, in order to decrease the total number of used priorities, all non-accepting SCCs can be assigned any odd priority that is already required elsewhere instead of the one suggested by the algorithm.
Further, one can adopt optimizations of Spot as well as consider optimizations taking the automaton topology more into account.
The whole tool-chain will then be integrated into Rabinizer.
Finally, in order to estimate the effect on LTL synthesis more precisely, we shall link our tool chain to parity-game solvers and apply it to realistic case studies.

\vfill
\pagebreak
\bibliographystyle{alpha}
\bibliography{refs}

\newpage
\input{appendix}

\end{document}

%% file: intro.tex

\section{Introduction}

Constructing correct-by-design systems from specifications given in linear temporal logic (LTL) \cite{DBLP:conf/focs/Pnueli77} is a classical problem \cite{DBLP:conf/popl/PnueliR89}, called \emph{LTL synthesis}.
The automata-theoretic solution to this problem is to translate the LTL formula to a deterministic automaton and solve the corresponding game on the automaton. 
Although different kinds of automata can be used, a reasonable choice would be parity automata (DPA) due to the practical efficiency of parity game solvers \cite{DBLP:conf/atva/FriedmannL09,michael} and the fact they allow for optimal memoryless strategies.
The bottleneck is thus to create a reasonably small DPA.
The classical way to transform LTL formulae into DPA is to first create a non-deterministic Büchi automaton (NBA) and then determinize it, as implemented in \texttt{ltl2dstar} \cite{ltl2dstar}.
Since determinization procedures \cite{DBLP:conf/lics/Piterman06,DBLP:conf/fossacs/Schewe09} based on Safra's construction \cite{DBLP:conf/focs/Safra88} are practically inefficient, many alternative approaches to LTL synthesis arose, trying to avoid determinization and/or focusing on fragments of LTL, e.g. \cite{DBLP:conf/focs/KupfermanV05,DBLP:conf/vmcai/PitermanPS06,DBLP:journals/tocl/AlurT04}.
However, new results on translating LTL directly and efficiently into deterministic automata \cite{DBLP:conf/cav/KretinskyE12,DBLP:conf/cav/EsparzaK14} open new possibilities for the automata-theoretic approach.
Indeed, tools such as Rabinizer \cite{DBLP:conf/atva/KomarkovaK14} or LTL3DRA \cite{DBLP:conf/atva/BabiakBKS13} can produce practically small deterministic Rabin automata (DRA).
Consequently, the task is to efficiently transform DRA into DPA, which is the aim of this paper.

Transformations of deterministic automata into DPA are mostly based on \emph{appearance records} \cite{DBLP:conf/stoc/GurevichH82}.
For instance, for deterministic Muller automata, we want to track which states appear infinitely often and which do not.
In order to do that, the \emph{state appearance record} keeps a permutation of the states, ordered according to their most recent visits, see e.g. \cite{DBLP:conf/dagstuhl/Schwoon01}.
In contrast, for deterministic Streett automata (DSA) we only want to track which \emph{sets} of states are visited infinitely often and which not.
Consequently, \emph{index appearance record} (IAR) keeps a permutation of these sets of interest instead, which are typically very few. Such a transformation has been given first in \cite{DBLP:conf/stoc/Safra92} from DSA to DRA only (not DPA, which is a subclass of DRA).
Fortunately, this construction can be further modified into a transformation of DSA to DPA, as shown in \cite{loding-thesis}.

Since 1) DRA and DSA are syntactically the same, recognizing the complement languages of each other, and 2) DPA can be complemented without any cost, one can apply the IAR of \cite{loding-thesis} to DRA, too.
However, we design another IAR, which is more natural from the DRA point of view, as opposed to the DSA perspective taken in \cite{loding-thesis}.
This is in spirit more similar to a sketch of a construction suggested in \cite{exercises}.
Surprisingly, we have found that the DRA perspective yields an algorithm producing considerably smaller automata than the DSA perspective.

Our contribution in this paper is as follows:
\begin{itemize}
	\item We provide an IAR construction transforming DRA to DPA.
	\item We present optimizations applicable to all appearance records.
	\item We evaluate all the unoptimized and optimized versions of our IAR and the IAR of \cite{loding-thesis} experimentally, in comparison to the procedure implemented in GOAL \cite{DBLP:conf/cav/TsaiTH13}.
	\item We compare our approach LTL$\xrightarrow{\text{Rabinizer}}$DRA$\xrightarrow{\text{optimized IAR}}$DPA to the state-of-the-art translation of LTL to DPA by Spot 2.1 \cite{spot2}, which mixes the construction of \cite{DBLP:journals/fuin/Redziejowski12} with some optimizations of \texttt{ltl2dstar} \cite{ltl2dstar} and of their own.
	The experiments show that for more complex formulae our method produces smaller automata.
\end{itemize}

%% file: defs.tex
\section{Preliminaries on $\omega$-automata} \label{section:definitions}
\vspace{-1pt}

We recall basic definitions of $\omega$-automata and establish some notation.

\vspace{-1pt}
\subsection{Alphabets and words} \label{section:definitions:alphabets}

An \emph{alphabet} is any finite set $\Sigma$.
The elements of $\Sigma$ are called \emph{letters}.
A \emph{word} is a (possibly infinite) sequence of letters.
The set of all infinite words is denoted by $\Sigma^\omega$.
A set of words $\Language \subseteq \Sigma^\omega$ is called \emph{(infinite) language}.
The $i$-th letter of a word $w \in \Sigma^\omega$ is denoted by $w_i$, \abbrevie $w = w_0 w_1 \dots$.

\vspace{-1pt}
\subsection{Transition systems} \label{section:definitions:transitions}
\vspace{-1pt}

A \emph{deterministic transition system} (DTS) $\detTransSystem$ is given by a tuple $(Q, \Sigma, \delta, q_0)$ where $Q$ is a set of states, $\Sigma$ is an alphabet, $\delta$ is a \emph{transition function} $\delta : Q \times \Sigma \to Q$ which may be partial (due to technical reasons) and $q_0 \in Q$ is the \emph{initial state}.
The transition function induces the \emph{set of transitions} $\Delta = \set*{\trans{q}{a}{q'} \mid q \in Q, a \in \Sigma, q' = \fun{\delta}{p, a}}$.
For a transition $t = \trans{q}{a}{q'} \in \Delta$ we say that $t$ \emph{starts at} $q$, \emph{moves under} $a$ and \emph{ends in} $q'$.
A sequence of transitions $\rho$ is a \emph{run} of a DTS $\detTransSystem$ on a word $w \in \Sigma^\omega$ if $\rho_0$ starts at $q_0$, $\rho_i$ moves under $w_i$ for each $i \geq 0$ and $\rho_{i+1}$ starts at the same state as $\rho_i$ ends for each $i \geq 0$.
We write $\fun{\detTransSystem}{w}$ to denote the unique run of $\detTransSystem$ on $w$, if it exists.
A transition $t$ \emph{occurs} in $\rho$ if there is some $i$ with $\rho_i = t$.
By $\fun{\Inf}{\rho}$ we denote the set of all transitions occurring infinitely often in $\rho$.
Additionally, we extend $\Inf$ to words by defining $\fun{\Inf_{\detTransSystem}}{w} = \fun{\Inf}{\fun{\detTransSystem}{w}}$ if $\detTransSystem$ has a run on $w$.
If $\detTransSystem$ is clear from the context, we write $\fun{\Inf}{w}$ for $\fun{\Inf_{\detTransSystem}}{w}$.

\vspace{-1pt}
\subsection{Acceptance conditions and $\omega$-automata} \label{section:definitions:automata}
\vspace{-1pt}

An \emph{acceptance condition} for $\detTransSystem$ is a positive Boolean formula over the formal variables  $V_\Delta = \set{\fun{\Inf}{T}, \fun{\Fin}{T} \mid T \subseteq \Delta}$.
Acceptance conditions are interpreted over runs as follows.
Given a run $\rho$ of $\detTransSystem$ and such an acceptance condition $\alpha$, we consider the truth assignment that sets the variable $\fun{\Inf}{T}$ to true iff $\rho$ visits (some transition of) $T$ infinitely often, \abbrevie $\fun{\Inf}{\rho} \intersection T \neq \emptyset$.
Dually, $\fun{\Fin}{T}$ is set to true iff $\rho$ visits every transition in $T$ finitely often, \abbrevie $\fun{\Inf}{\rho} \intersection T = \emptyset$.
A run $\rho$ satisfies $\alpha$ if this truth-assignment evaluates $\alpha$ to true.

A \emph{deterministic $\omega$-automaton} over $\Sigma$ is a tuple $\Automaton = (Q, \Sigma, \delta, q_0, \alpha)$, where $\tuple*{Q, \Sigma, \delta, q_0}$ is a DTS and $\alpha$ is an acceptance condition for it.
An automaton $\Automaton$ \emph{accepts} a word $w \in \Sigma^\omega$ if the run of the automaton on $w$ satisfies $\alpha$.
The language of $\Automaton$, denoted by $\Language(\Automaton)$, is the set of words accepted by $\Automaton$.
An acceptance condition $\alpha$ is a
\begin{itemize}
	\item \emph{Rabin condition $\set*{\tuple*{F_i, I_i}}_{i=1}^k$} if $\alpha = \boolOr_{i=1}^k (\fun{\Fin}{F_i} \booland \fun{\Inf}{I_i})$.
	Each $\tuple{F_i, I_i}$ is called a \emph{Rabin pair}, where the $F_i$ and $I_i$ are called the \emph{prohibited set} and the \emph{required set} respectively.
	\item \emph{generalized Rabin condition $\set*{\tuple*{F_i,\set*{I_i^j}_{j=1}^{k_i}}}_{i=1}^k$} if the acceptance condition is of the form $\alpha = \boolOr_{i=1}^n (\fun*{\Fin}{F_i} \booland \boolAnd_{j=1}^{k_i} \fun*{\Inf}{I_j^k)}$.
	This generalizes the Rabin condition, where each $k_i=1$.
	Furthermore, every generalized Rabin automaton can be de-generalized into an equivalent Rabin automaton, which however may incur an exponential blow-up \cite{DBLP:conf/cav/KretinskyE12}.
	\item \emph{Streett condition $\set*{\tuple*{F_i, I_i}}_{i=1}^k$} if $\alpha = \boolAnd_{i=1}^k (\fun{\Inf}{F_i} \boolor \fun{\Fin}{I_i})$.
	Note that the Streett condition is exactly the negation of the Rabin condition and thus an automaton with a Rabin condition can be interpreted as a Streett automaton recognizing exactly the complement language.
	\item \emph{Rabin chain condition $\set*{\tuple*{F_i, I_i}}_{i=1}^k$} if it is a Rabin condition and $F_1 \subseteq I_1 \subseteq \cdots \subseteq F_k \subseteq I_k$.
	A Rabin chain condition is equivalent to a \emph{parity condition}, specified by a priority assignment $\lambda : \Delta \to \Naturals$.
	Such a parity condition is satisfied by a run $\rho$ iff the maximum priority of all infinitely often visited transitions $\max \set*{\fun{\lambda}{q} \mid q \in \fun{\Inf}{\rho}}$ is even.
\end{itemize}
A deterministic Rabin, generalized Rabin, Street or parity automaton is a deterministic $\omega$-automaton with an acceptance condition of the corresponding kind.
In the rest of the paper we use the corresponding abbreviations DRA, DGRA, DSA and DPA.

Furthermore, given a DRA with an acceptance set $\set{\tuple{F_i, I_i}}_{i=1}^k$ and a word $w \in \Sigma^\omega$, we write $\ProhibitedInf = \set{F_i \mid F_i \intersection \fun{\Inf}{w} \neq \emptyset}$ and $\RequiredInf = \set{I_i \mid I_i \intersection \fun{\Inf}{w} \neq \emptyset}$ to denote the set of all infinitely often visited prohibited and required sets, respectively.

%% file: iar.tex
\section{Index appearance record} \label{section:index_appearance_record}

In order to translate (state-based acceptance) Muller automata to parity automata, a construction called \emph{latest appearance record} has been devised\footnote{Originally, it appeared in an unpublished report of McNaughton under the name \enquote{order vector with hit}}.
In essence, the constructed state space consists of permutations of all states in the original automaton.
In each transition, the state which has just been visited is moved to the front of the permutation.
From this, one can deduce the set of all infinitely often visited states by investigating which states change their position in the permutation infinitely often along the run of the word.
Such a constraint can be encoded as parity condition.

However, this approach comes with a very fast growing state space, as the amount of permutations grows exponentially.
Moreover, applying this idea to transition based acceptance leads to even faster growth, as there usually are a lot more transitions than states.
In contrast to Muller automata, the exact set of infinitely often visited transitions is not needed to decide acceptance of a word by a Rabin automaton.
It is sufficient to know which of the prohibited and required \emph{sets} are visited infinitely often.
Hence, \emph{index appearance record} uses the indices of the Rabin pairs instead of particular states in the permutation construction.
This provides enough information to decide acceptance.

We introduce some formalities regarding permutations: For a given $n \in \Naturals$, we use $\Pi^n$ to denote the set of all permutations of $N = \set{1, \dots, n}$, \abbrevie the set of all bijective functions $\pi : N \to N$.
We identify $\pi$ with its canonical representation as a vector $\tuple{\pi(1), \dots, \pi(n)}$.
In the following, we will often say \enquote{the position of $F_i$ in $\pi$} or similar to refer to the position of $i$ in a particular $\pi$, \abbrevie $\fun{\pi^{-1}}{i}$.
With this, we define our variant of the index appearance record construction.
Note that in contrast to previous constructions, ours is transition based, which also has a positive effect on the size of the produced automata, as discussed in our experimental results.
\begin{definition}[Transition-based index appearance record for Rabin automata] \label{def:IAR}
	Let $\mathcal{R} = \tuple*{Q, \Sigma, \delta, q_0, \set*{\tuple*{F_i, I_i}}_{i=1}^k}$ be a Rabin automaton.
	Then the \emph{index appearance record automaton} $\fun{\IAR}{\mathcal{R}} = \tuple*{\tilde{Q}, \Sigma, \tilde{\delta}, \tilde{q}_0, \lambda}$ is defined as the parity automaton with
	\begin{itemize}
		\item $\tilde{Q} = Q \times \Pi^k$.
		\item $\tilde{q}_0 = \tuple{q_0, \tuple*{1, \dots, k}}$.
		\item $\fun{\tilde{\delta}}{\tuple{q, \pi}, a} = \tuple{\fun{\delta}{q, a}, \pi'}$ where $\pi'$ is the permutation obtained from $\pi$ by moving all indices of \emph{prohibited} sets visited by the transition $t = \trans{q}{a}{\fun{\delta}{q, a}}$ to the front.
		Formally, let $\mathsf{Move} = \set*{i \mid t \in F_{\fun{\pi}{i}}}$ be the set of positions of currently visited prohibited sets.
		If $\mathsf{Move} = \emptyset$, define $\pi' = \pi$, otherwise let $n = \cardinality{\mathsf{Move}}$ and $\mathsf{Move} = \set*{i_1, \dots, i_n}$.
		With this
		\begin{equation*}
			\pi'(j) = \begin{dcases*}
				i_j & if $j \leq n$\\
				\fun\pi{j - n + \cardinality*{\set*{i \in \mathsf{Move} \mid i \leq j}}} & otherwise.
			\end{dcases*}
		\end{equation*}
		\item To define the priority assignment, we first introduce some auxiliary notation. For a transition $\tilde{t} = \trans{\tuple{q, \pi}}{a}{\tuple{q', \pi'}}$ and its corresponding transition $\trans{q}{a}{q'}$ in the original automaton, let
		\[\fun{\maxInd}{\tilde{t}} = \fun{\max}{\set{\fun{\pi^{-1}}{i} \mid t \in F_i \union I_i} \union \set*{0}}\] be the maximal position of \emph{acceptance pair} in $\pi$ visited by $t$ (or $0$ if none is visited).
		Using this, define the priority assignment as follows:
		\begin{equation*}
			\fun{\lambda}{\tilde{t}} := \begin{dcases*}
				1                                     & if $\fun{\maxInd}{\tilde{t}} = 0$, \\
				2 \cdot \fun*{\maxInd}{\tilde{t}}     & if $t \in I_{\fun*{\pi}{\fun*{\maxInd}{\tilde{t}}}} \setminus F_{\pi(\fun*{\maxInd}{\tilde{t}})}$ \\
				2 \cdot \fun*{\maxInd}{\tilde{t}} + 1 & otherwise, i.e. if $t \in  F_{\pi(\fun*{\maxInd}{\tilde{t}})}$.
			\end{dcases*}
		\end{equation*}
	\end{itemize}
\end{definition}
When a transition visits multiple prohibited sets, they can be moved to the front of the appearance record in arbitrary order.
As an optimization we choose existing states as successors whenever possible.

Before formally proving correctness, \abbrevie that $\fun{\IAR}{\mathcal{R}}$ recognizes the same language as $\mathcal{R}$, we provide a small example in \cref{fig:iar_example} and explain the general intuition behind the construction.
For a given run, all prohibited sets which are visited infinitely often will eventually be \enquote{in front} of all those only seen finitely often: After some finite number of steps, none of the finitely often visited ones will be seen any more.
Taking another sufficiently large amount of steps, every infinitely often visited set has been seen again and all their indices have been moved to the front.
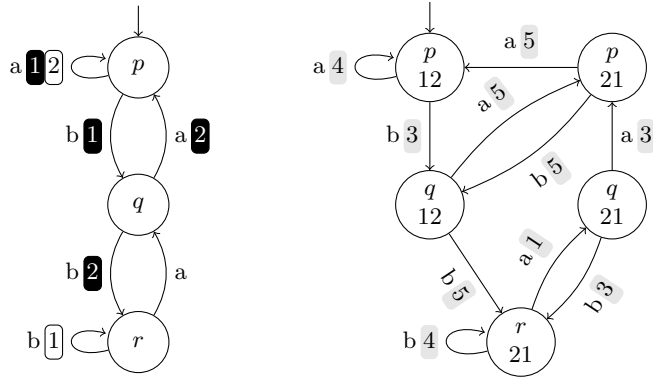
\begin{figure}[t]
	\centering
	\begin{tikzpicture}[auto,node distance=1cm,initial text=]
		\tikzstyle{every state}=[align=center]
		\node[state,initial above] (p)                {$p$};
		\node[state]               (q)   [below=of p] {$q$};
		\node[state]               (r)   [below=of q] {$r$};
		
		\node[state,initial above] (p12) [right=3cm of p]     {$p$\\$12$};
		\node[state]               (p21) [right=1.5cm of p12]  {$p$\\$21$};
		\node[state]               (q12) [right=3cm of q]     {$q$\\$12$};
		\node[state]               (q21) [right=1.5cm of q12]  {$q$\\$21$};
		\node                      (r_anchor) at ($(q12)!0.5!(q21)$) {};
		\node[state]               (r21) at (r_anchor |- r)  {$r$\\$21$};

		\path[->]
		(p)   edge [bend right,swap] node {b\,\edgeprohibited{1}} (q)
		      edge [loop left]       node {a\,\edgeprohibited{1}\edgerequired{2}} (p)
		(q)   edge [bend right,swap] node {a\,\edgeprohibited{2}} (p)
		      edge [bend right,swap] node {b\,\edgeprohibited{2}} (r)
		(r)   edge [bend right,swap] node {a} (q)
		      edge [loop left]       node {b\,\edgerequired{1}} (r)

		(p12) edge [loop left]       node {a\,\edgepriority{4}} (p12)
		      edge [swap]            node {b\,\edgepriority{3}} (q12)
		(q12) edge [bend left=15]    node[sloped,anchor=south,auto=false] {a\,\edgepriority{5}} (p21)
		      edge [swap]            node[sloped,anchor=north,auto=false] {b\,\edgepriority{5}} (r21)
		(r21) edge [loop left]       node {b\,\edgepriority{4}} (r21)
		      edge [bend left=15]    node[sloped,anchor=south,auto=false] {a\,\edgepriority{1}} (q21)
		(q21) edge [bend left=15]    node[sloped,anchor=north,auto=false] {b\,\edgepriority{3}} (r21)
		      edge [swap]            node {a\,\edgepriority{3}} (p21)
		(p21) edge [swap]            node {a\,\edgepriority{5}} (p12)
		      edge [bend left=15]    node[sloped,anchor=north,auto=false] {b\,\edgepriority{5}} (q12)
		;
	\end{tikzpicture}
	
	\caption{An example DRA and its resulting IAR DPA.
	For the Rabin automaton, a number in a white box next to a transition indicates that this transition is a required one of that Rabin pair.
	A black shape dually indicates membership in the corresponding prohibited set.
	For example, with $t = \trans{p}{a}{p}$ we have $t \in F_1$ and $t \in I_2$.
	In the IAR construction, we shorten the notation for permutations to save space, so $p, 12$ corresponds to $\tuple{p, \tuple{1, 2}}$.
	The priority of a transition is written next to the transitions letter.}
	\label{fig:iar_example}
\end{figure}

\begin{lemma}\label{stm:iar_inf_visited_indices}
	Let $w \in \Sigma^\omega$ be a word on which $\fun{\IAR}{\mathcal{R}}$ has a run $\tilde{\rho}$.
	Then, the positions of all finitely often visited prohibited sets stabilize after a finite number of steps, \abbrevie their positions are identical in all infinitely often visited states.
	Moreover, for any $i, j$ with $F_i \in \ProhibitedInf$, $F_j \notin \ProhibitedInf$ we have that the position of $F_i$ is smaller than the position of $F_j$ in every infinitely often visited state.
\end{lemma}
\begin{proof}
	The position of any $F_i$ only changes in two different ways:
	\begin{itemize}
		\item Either $F_i$ itself has been visited and thus is moved to the front,
		\item or some $F_{i'}$ with a position greater than the one of $F_i$ has been visited and is moved to the front, increasing the position of $F_i$.
	\end{itemize}
	Let $\rho$ be the run of $\mathcal{R}$ on $w$.
	(We prove the existence of such a run in \cref{stm:iar_star_simulates_rabin_run}.)
	Assume that $F_i$ is visited finitely often in some run $\rho$, \abbrevie there is a step in the run from which on $F_i$ is never visited again.
	As the amount of positions is bounded, the second case may only occur finitely often after this step and the position of $F_i$ eventually remains constant.
	As $F_i$ was chosen arbitrarily, we conclude that all finitely often visited $F_i$ are eventually moved to the right and remain on their position.
	Trivially, all infinitely often visited $F_i$ move to the left, proving the claim.\qed
\end{proof}
As an immediate consequence we see that if some transition $(q,a,q') \in F_i$ is visited infinitely often, then every $F_j$ with a smaller position than $F_i$ in $q$ is also visited infinitely often:
\begin{corollary}\label{stm:smaller_indices_visited_inf_often}
	Let $\tilde{t} \in \Inf_{\IAR(\mathcal{R})}(w)$ 
	be an infinitely often visited transition with its corresponding transition $t \in F_{\fun{\pi}{i}}$ for some $i$.
	Then $\Forall{j \leq i}{F_{\fun{\pi}{j}} \in \ProhibitedInf}$.
\end{corollary}
Looking back at the definition of the priority function, the central idea of correctness can be outlined as follows.
For every $I_i$ which is visited infinitely often we can distinguish two cases:
\begin{itemize}
	\item $F_i$ is visited finitely often.
	Then the position of the pair is greater than the one of every $F_j \in \ProhibitedInf$.
	Hence the priority of every transition $\tilde{t}$ with corresponding transition $t \in I_i$ is both even and bigger than every odd priority seen infinitely often along the run.
	\item $F_i$ is visited infinitely often, \abbrevie after each visit of $I_i$, $F_i$ is eventually visited.
	As argued in the proof of \cref{stm:iar_inf_visited_indices}, the position of $F_i$ can only increase until it is visited again.
	Hence every visit of $I_i$ which yields an even parity is followed by a visit of $F_i$ yielding an odd parity which is strictly greater.
\end{itemize}
Using this intuition, we formally show correctness of the construction in \cref{section:appendix:iar_correctness}.
\begin{theorem}\label{stm:iar_correctness}
	For any DRA $\mathcal{R}$ we have that $\Language(\fun{\IAR}{\mathcal{R}}) = \Language(\mathcal{R})$.
\end{theorem}
\begin{proposition}[Complexity]\label{stm:complexity}
	For every DRA $\mathcal{R}$ with $n$ states and $k$ Rabin pairs, the constructed automaton $\fun{\IAR}{\mathcal{R}}$ has at most $n \cdot k!$ states and $2k+1$ priorities.
\end{proposition}
Moreover, using the \cite{DBLP:conf/fsttcs/Loding99}, one can show that 
this is essentially optimal.
There exists a family $\sequence{\Language_n}_{n \geq 2}$ of languages such that for every $n$ the language $L_n$ can be recognized by a DRA with $O(n)$ states and $O(n)$ pairs, but cannot be recognized by a DPA with less than $n!$ states.
For details, see Appendix~\ref{app:complexity}.

\begin{remark}[Comparison to previous IAR]\label{rem:streett}
	Our construction is similar to the index appearance record of \cite{loding-thesis} in that it keeps the information about the current state and a permutation of pairs, implementing the appearance record.
	However, from the point of view of Streett automata, it is very natural to keep two pointers into the permutation, indicating the currently extreme positions of both types of sets in the accpetance condition.
	Indeed, this way we can keep track of all conjuncts of the form $\Inf(I_j)\implies\Inf(F_j)$.
	This is also the approach that \cite{loding-thesis} takes.
	In contrast, we have no pointers at all.
	From the Rabin point of view, it is more natural to keep track of the prohibited sets only and the respective pointer is hidden in the information about the current state \emph{together} with the current permutation.
	Additionally, the pointer for the required set is hidden into the acceptance status of transitions. In the transition-based setting, it is not necessary to remember the visit of a required set in the state-space; it is sufficient to emit the respective priority upon seeing this \emph{during} the transition when we know both the source and target states. 
	The absence of these pointers results in better performance.
\end{remark}

\begin{remark}[Using IAR for DGRA]
	The straightforward way to translate a DGRA to DPA is to first de-generalize the DGRA and then apply the presented $\IAR$ construction.
	However, one can also apply the IAR idea to directly translate from DGRA to DPA: Instead of only tracking the pair indices, one could incorporate all $F_i$ and $I_i^j$ into the appearance permutation.
	With the same reasoning as above, a parity condition can be used to decide acceptance.
	
	This approach yields a correct algorithm, but compared to de-generalization combined with IAR, the state space grows much larger.
	Indeed, given a DGRA with $n$ states and $k$ accepting pairs with $l_i$ required sets each, the de-generalized DRA has at most $n \cdot \prod_{i=1}^k l_i$ states and $k$ pairs, hence the resulting parity automaton has at most $k! \cdot n \cdot \prod_{i=1}^k l_{i}$ states and $2k+1$ priorities.
	Applying the mentioned specific construction gives $n \cdot (\sum_{i=1}^k (l_i +1))!$ states and $2 \cdot (\sum_{i=1}^k (l_i + 1)) + 1$ priorities.
	A simple induction on $k$ suffices to show that the worst case upper bound for the specific construction is always larger.
	We conjecture that this behaviour also shows in real-world applications.
\end{remark}

%% file: opt.tex
\section{Optimizations} \label{section:optimizations}

In general, many states generated by the IAR procedure are often superfluous and could be omitted. 
In the following, we present 
several optimizations of our construction, which aim to do so. Moreover, these optimizations can be applied also to the IAR construction of \cite{loding-thesis} and in a slighly adjusted way also to the standard SAR \cite{DBLP:conf/dagstuhl/Schwoon01}.
Further, although the optimizations are transition-based, they can be of course easily adapted to the state-based setting.
Due to space constraints, the correctness proofs can be found in \cref{section:appendix:iar_star_correctness}.

\subsection{Choosing an initial permutation} \label{section:optimizations:choosing_inital_permutation}

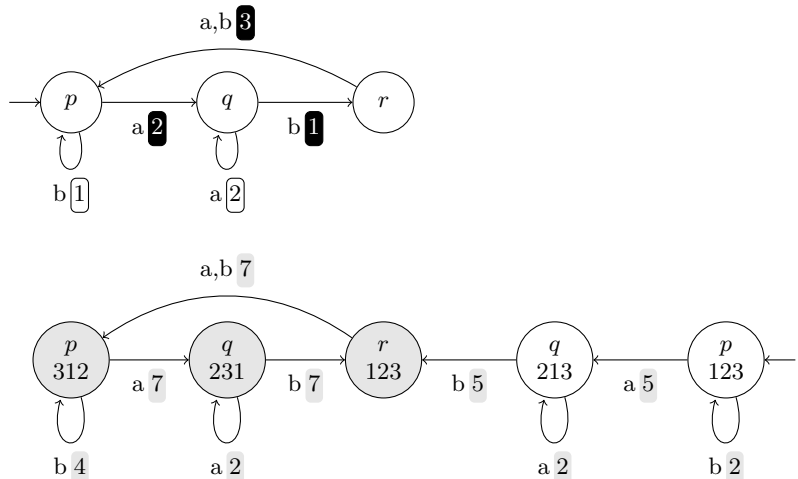
\begin{figure}[t]
	\centering
	\begin{tikzpicture}[auto,initial text=,node distance=1.25cm]
	\tikzstyle{every state}=[align=center]
	\node[state,initial left] (p) {$p$};
	\node[state,right=of p]   (q) {$q$};
	\node[state,right=of q]   (r) {$r$};
	
	\path[->]
	(p) edge [swap]            node {a\,\edgeprohibited{2}} (q)
	    edge [loop below]      node {b\,\edgerequired{1}} (p)
	(q) edge [swap]            node {b\,\edgeprohibited{1}} (r)
	    edge [loop below]      node {a\,\edgerequired{2}} (q)
	(r) edge [bend right,swap] node {a,b\,\edgeprohibited{3}} (p)
	;
	
	\node[state,below=2.5cm of p,fill=black!10] (p312) {$p$\\$312$};
	\node[state,below=2.5cm of q,fill=black!10] (q231) {$q$\\$231$};
	\node[state,below=2.5cm of r,fill=black!10] (r123) {$r$\\$123$};
	\node[state,right=of r123]                  (q213) {$q$\\$213$};
	\node[state,right=of q213,initial right]    (p123) {$p$\\$123$};
	
	\path[->]
	(p312) edge [swap]               node {a\,\edgepriority{7}} (q231)
	       edge [loop below]         node {b\,\edgepriority{4}} (p123)
	(q231) edge [loop below]         node {a\,\edgepriority{2}} (q231)
	       edge [swap]               node {b\,\edgepriority{7}} (r123)
	(r123) edge [bend right=35,swap] node {a,b\,\edgepriority{7}} (p312)
	(q213) edge [loop below]         node {a\,\edgepriority{2}} (q213)
	       edge []                   node {b\,\edgepriority{5}} (r123)
	(p123) edge []                   node {a\,\edgepriority{5}} (q213)
	       edge [loop below]         node {b\,\edgepriority{2}} (p123)
	;
	\end{tikzpicture}
	\caption{Example of a suboptimal initial permutation, using the same notation as in \cref{fig:iar_example}.
	Only the shaded states are constructed when choosing a better initial permutation.}
	\label{fig:example_bad_initial_permutation}
\end{figure}

The first observation is that the arbitrary choice of $\tuple{1, \dots, k}$ as initial permutation can lead to suboptimal results.
It may happen that several states of the resulting automaton are visited at most once by every run before some \enquote{recurrent} permutation is reached.
These states enlarge the state-space unnecessarily, as demonstrated in \cref{fig:example_bad_initial_permutation}.
Indeed, when choosing $\tuple*{p, \tuple*{3, 1, 2}}$ instead of $\tuple*{p, \tuple*{1, 2, 3}}$ as the initial state in the example, only the shaded states are built during the construction, while the language of the resulting automaton is still equal to that of the input DRA.

We overload the $\IAR$ algorithm to be parametrized by the starting permutation, \abbrevie we write $\fun{\IAR}{\mathcal{R}, \pi_0}$ to denote the IAR construction applied to the DRA $\mathcal{R}$ starting with permutation $\pi_0$.
\begin{theorem} \label{stm:iar_independence_of_initial_permutation}
	For an arbitrary Rabin automaton $\mathcal{R}$ with $k$ pairs we have that $\Language(\fun{\IAR}{\mathcal{R}}) = \Language(\fun{\IAR}{\mathcal{R}, \pi_0})$ for all $\pi_0 \in \Pi^k$.
\end{theorem}
How to choose a \enquote{good} initial permutation is deferred to \cref{section:optimizations:pick_perm}, as it is intertwined with the algorithm presented in the following section.

\subsection{SCC decomposition}

Acceptance of a word by an $\omega$-automaton only depends on the set of states visited infinitely often by its run.
This set of states is \emph{strongly connected} on the underlying graph structure, \abbrevie starting from any state in the set, any other state can be reached with finitely many steps.
In general, any strongly connected set belongs to exactly one \emph{strongly connected component} (SCC).
Therefore, for a fixed SCC, only the Rabin pairs with required sets intersecting this SCC are relevant.

Using this we can restrict ourselves to the Rabin pairs that can possibly accept in that SCC while processing it.
This reduces the number of indices we need to track in the appearance record for each SCC, which can lead to significant savings.

For readability, we introduce some abbreviations.
Given a DRA $\mathcal{R} = \linebreak[4] \tuple*{Q, \Sigma, \delta, q_0, \set*{\tuple{F_i, I_i}}_{j=1}^k}$ and a set of states $S \subseteq Q$ we write $\delta \upharpoonright S : S \times \Sigma \to S$ to denote the restriction of $\delta$ to $S$, \abbrevie $\fun{\delta \upharpoonright S}{q, a} = \fun{\delta}{q, a}$ if $\fun{\delta}{q, a} \in S$ and undefined otherwise.
Analogously, we define $\Delta \upharpoonright S = \Delta \intersection S \times \Sigma \times S$ as the set of transitions in the restricted automaton.
Consequently, we define the restriction of the whole automaton $\mathcal{R}$ to the set of states $S$ using $q \in S$ as initial state by
\begin{equation*}
	\mathcal{R} \upharpoonright_q S = \tuple*{S, \Sigma, \delta \upharpoonright S, q, \set*{\tuple*{F_i \intersection (\Delta \upharpoonright S), I_i \intersection (\Delta \upharpoonright S)} \mid I_i \intersection (\Delta \upharpoonright S) \neq \emptyset}}.
\end{equation*}
Furthermore, we call a SCC of an automaton \emph{transient}, if it is a singleton set without a self-loop.
This means that it is visited at most once by any run and it is not of interest for acceptance.
Finally, we use $\varepsilon$ to denote the \enquote{empty} permutation (of length $0$).

Using this notation, we describe the optimized IAR construction, denoted IAR$^*$ in \cref{alg:iar_star}.
The algorithm decomposes the DRA into its SCCs, applies the formerly introduced IAR procedure to each sub-automaton separately and finally connects the resulting DPAs back together.
\begin{algorithm}[p]
	\DontPrintSemicolon
	\SetKwInOut{Input}{Input}\SetKwInOut{Output}{Output}
	
	\Input{A DRA $\mathcal{R} = \tuple*{Q, \Sigma, \delta, q_0, \set*{\tuple*{F_i, I_i}}_{j=1}^k}$}
	\BlankLine
	\Output{A DPA recognizing the same language as $\mathcal{R}$}
	\BlankLine
	
	$Q^* \gets \set*{}$, $\delta^* \gets \set*{}$, $\lambda^* \equiv 1$\;
	\ForEach{SCC $S$ in $\mathcal{R}$}{
		\uIf(\tcp*[h]{SCC not relevant}){$S$ transient or $\set*{i \mid I_i \intersection \Delta \upharpoonright S \neq \emptyset} = \emptyset$}{ \label{alg:iar_star:line:scc_distinction_if}
			Add $S \times \set*{\varepsilon}$ to $Q^*$\; \label{alg:iar_star:line:update_q_star_in_trivial_scc}
			\ForEach{$q \in S$, $a \in \Sigma$ such that $\fun{(\delta \upharpoonright S)}{q, a}$ is defined}{
				Let $q' = \fun{\delta}{q, a}$\;
				Set $\fun*{\delta^*}{\tuple*{q, \varepsilon}, a} = \tuple*{q', \varepsilon}$ and $\fun*{\lambda^*}{\trans*{\tuple*{q, \varepsilon}}{a}{\tuple*{q', \varepsilon}}} = 1$\label{alg:iar_star:line:update_delta_star_in_trivial_scc}
			}
		}\Else(\tcp*[h]{SCC relevant, apply IAR to the sub-automaton}){
			Pick a starting state $q \in S$\;
			$\tuple{Q_S, \Sigma, \delta_S, \tuple*{q, \pi}, \lambda_S} \gets \fun*{\IAR}{\mathcal{R} \upharpoonright_q S, \fun{\pickPerm}{q, S}}$\;\label{alg:iar_star:line:do_scc_iar}
			Update $Q^*$, $\delta^*$ and $\lambda^*$ with $Q_S$, $\delta_S$ and $\lambda_S$, respectively\label{alg:iar_star:line:update_variables_after_iar}
		}
	}
	\tcp{Connect all SCCs}
	\ForEach{$\tuple*{q, \pi} \in Q^*$ and $a \in \Sigma$ \abbrevst $q' = \fun{\delta}{q, a}$ in different SCC of $\mathcal{R}$ than $q$}{ \label{alg:iar_star:line:for_loop_inter_scc_connection}
		Pick a $\pi'$ with $\tuple*{q', \pi'} \in Q^*$\;\label{alg:iar_star:line:pick_pi_prime_for_inter_scc_connection}
		Set $\fun{\delta^*}{\tuple*{q, \pi}, a} = \tuple*{q', \pi'}$ \label{alg:iar_star:line:update_delta_star_with_inter_scc_connection}
	}
	\caption{The optimized $\IAR$ construction $\IAR^*$}
	\label{alg:iar_star}
\end{algorithm}

\begin{figure}[p]
	\centering
	\begin{tikzpicture}[auto,node distance=1cm,initial text=]
		\tikzstyle{every state}=[align=center,minimum size=0.9cm]
		\node[state,initial above]         (p)              {$p$};
		\node[state]                       (q) [below=of p] {$q$};
		\node[state]                       (r) [below=of q] {$r$};
		\node[text width=4cm,align=center] (dra_text) [below=0.5cm of r] {Example DRA,\\SCCs shaded};
		
		\node[rectangle,rounded corners=10pt,draw=none,fill=black,fill opacity=0.1,fit=(p)] {};
		\node[rectangle,rounded corners=10pt,draw=none,fill=black,fill opacity=0.1,fit=(q) (r)] {};
		
		\path[->]
		(p) edge  [loop left] node {\,a\,\edgerequired{1}\edgeprohibited{2}} (p)
		    edge  []          node {b} (q)
		(q) edge  [bend left] node {\,a\,\edgeprohibited{2}} (r)
		    edge  [loop left] node {\,b\,\edgerequired{2}} (q)
		(r) edge  [loop left] node {\,a\,\edgeprohibited{1}} (r)
		    edge  [bend left] node {\,b\,\edgeprohibited{1}\edgerequired{2}} (q)
		;
		
		\node[state,initial above]         (p12) [right=2cm of p] {$p$\\$12$};
		\node[state]                       (p21) [right=of p12]   {$p$\\$21$};
		\node[state]                       (q12) at (p12 |- q) []  {$q$\\$12$};
		\node[state]                       (q21) [right=of q12]   {$q$\\$21$};
		\node[state]                       (r12) at (p12 |- r) []  {$r$\\$12$};
		\node[state]                       (r21) [right=of r12]   {$r$\\$21$};
		\node                              (dpa_text_anchor) at ($(r12)!0.5!(r21)+(0,0)$) {};
		\node[text width=4cm,align=center] (dpa_text) at (dpa_text_anchor |- dra_text) {Result of unoptimized $\IAR$ applied to the DRA};
		
		\path[->]
		(p12) edge []             node {a\,\edgepriority{5}} (p21)
		      edge []             node {b\,\edgepriority{1}} (q12)
		(p21) edge [loop right]   node {a\,\edgepriority{4}} (p21)
		      edge []             node {b\,\edgepriority{1}} (q21)
		(q12) edge [bend left=15] node[sloped,anchor=south,auto=false]  {a\,\edgepriority{5}} (r21)
		      edge [loop left]    node {b\,\edgepriority{4}} (q12)
		(q21) edge []             node {a\,\edgepriority{3}} (r21)
		      edge [loop right]   node {b\,\edgepriority{2}} (q21)
		(r12) edge [loop left]    node {a\,\edgepriority{3}} (r12)
		      edge []             node {b\,\edgepriority{4}} (q12)
		(r21) edge []             node {a\,\edgepriority{5}} (r12)
		      edge [bend left=15] node[sloped,anchor=north,auto=false]  {b\,\edgepriority{5}} (q12)
		;
		
		\node[state,initial above]         (o_p) [right=2cm of p21] {$p$\\$1$};
		\node[state]                       (o_q) at (o_p |- q) []   {$q$\\$2$};
		\node[state]                       (o_r) at (o_p |- r) []   {$r$\\$2$};
		\node[text width=4cm,align=center] (dpa_opt_text) at (o_r |- dpa_text) {Output of the optimized IAR$^*$.};
		
		\path[->]
		(o_p) edge [loop right] node {a\,\edgepriority{2}} (o_p)
		      edge []           node {b\,\edgepriority{1}} (o_q)
		(o_q) edge [bend left]  node {a\,\edgepriority{3}} (o_r)
		      edge [loop right] node {b\,\edgepriority{2}} (o_q)
		(o_r) edge [loop right] node {a\,\edgepriority{3}} (o_r)
		      edge [bend left]  node {b\,\edgepriority{3}} (o_q)
		;
	\end{tikzpicture}
	
	\caption{Example application of \cref{alg:iar_star}}
	\label{fig:example_iar_star}
\end{figure}
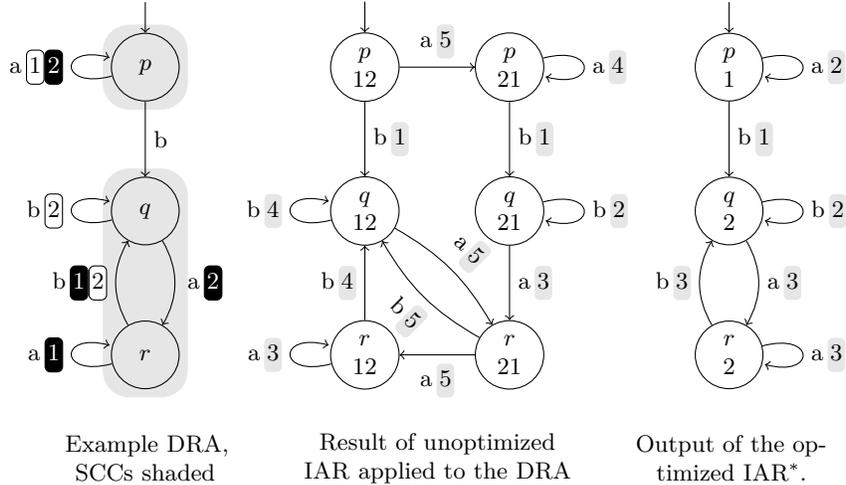

As we apply the IAR construction to each SCC separately, we have to choose the initial permutation for each state of those SCCs.
\cref{stm:iar_independence_of_initial_permutation} shows that for a particular initial state, correctness of IAR does not depend on the chosen permutation.
We therefore delegate the choice to a function $\pickPerm$ and prove correctness of the optimized algorithm independent of this function, allowing for further optimizations.
We present an optimal definition of $\pickPerm$ in the next subsection.

\cref{fig:example_iar_star} shows an example application and the obtained savings of the construction.
Pair $1$ is only relevant for acceptance in the SCC $\set*{p}$, but in the unoptimized construction it still changes the permutations in the part of the automaton constructed from $\set*{q, r}$, as \abbreveg the transition $\trans{r}{b}{q}$ is contained in $F_1$.
Similarly, pair $2$ is tracked in $\set*{p}$ while actually not being relevant.
The optimized version yields improvements in both state-space size and amount of priorities.
\begin{theorem}\label{stm:iar_star_correctness}
	For any DRA $\mathcal{R}$ we have that $\Language(\fun{\IAR^*}{\mathcal{R}}) = \Language(\mathcal{R})$, independent of $\pickPerm$.
\end{theorem}

\subsection{Optimal choice of the initial permutation} \label{section:optimizations:pick_perm}

In  \cref{fig:example_bad_initial_permutation} we provided a scalable example where the choice of the initial permutation can significantly reduce the size of the generated automaton.
In this subsection, we explain a procedure yielding a permutation which minimizes the state space of the automaton generated by $\IAR^*$.

First, we recall that $\pickPerm$ is only invoked when processing a particular (non-transient) SCC of the input automaton.
Consequently, we can restrict ourselves to only deal with Rabin automata forming a single SCC.
Let now $\mathcal{R}$ be such an automaton.
While $\fun{\IAR}{\mathcal{R}, \pi_0}$ may contain multiple SCCs, we show that it contains exactly one bottom SCC (BSCC), \abbrevie a SCC without outgoing edges.
Additionally, this BSCC is the only SCC which contains all states of the original automaton $\mathcal{R}$ in the first component of its states.
\begin{theorem} \label{stm:iar_star_subautomaton_bscc_unique}
	Let $\mathcal{R} = \tuple*{Q, \Sigma, \delta, q_0, \set*{\tuple*{F_i, I_i}}_{i=1}^k}$ be a Rabin automaton that is strongly connected.
	For a fixed $\pi_0 \in \Pi^k$, $\fun{\IAR}{\mathcal{R}, \pi_0}$ contains exactly one BSCC $S$ and for every SCC $S'$ we have that $S = S'$ iff $Q = \set*{q \mid \Exists \pi \in \Pi^k. \tuple*{q, \pi} \in S'}$.
	Furthermore the BSCCs for different $\pi_0$ are isomorphic.
\end{theorem}
The proof can be found in \cref{section:appendix:bscc_unique}.
This result makes defining an optimal choice of $\pickPerm$ straightforward.
By the theorem, there always is a BSCC of the same size, independent of $\pickPerm$.
If $\tuple*{q_0, \pi}$ is in the BSCC of some $\fun{\IAR}{\mathcal{R}, \pi_0}$, $\fun{\IAR}{\mathcal{R}, \pi}$ will generate the same BSCC and no other states.
Hence, we define $\fun{\pickPerm}{q, S}$ to return any permutation such that $\tuple*{q, \pi}$ lies in the corresponding BSCC.
As a trivial consequence of the theorem, this choice is optimal in terms of the state-space size of the generated automaton.
In our implementation, we start exploring the state space using an arbitrary initial permutation and then prune all states which do not belong into the respective BSCC.

%% file: exper.tex
\section{Experimental results}

In this section, we compare variants of our new approach to the established tools.
All of the benchmarks have been run on a Linux 4.4.3-gentoo x64 virtual machine with 3.0 GHz per core.
We implemented our construction as part of Rabinizer \cite{DBLP:conf/atva/KomarkovaK14} and used the 64 bit Oracle JDK 1.8.0\_102 as JVM for our experiments.

\subsection{DRA to DPA translation}

We present comparisons of different approaches to translate DRA into DPA.
As there are to our knowledge no \enquote{standard} DRA datasets for this kind of comparison, we use Spot's tool \texttt{randaut} to produce various Rabin automata.
All executions in this chapter ran with a time-out of five minutes.

We consider both our basic method IAR of \cref{section:index_appearance_record} and the optimized version IAR$^*$ of \cref{section:optimizations}.
We compare our methods to GOAL\footnote{\texttt{gc batch "\textbackslash\$nba = load -c HOAF /dev/stdin; \textbackslash\$dpa = convert -t dpw \textbackslash\$nba; save \textbackslash\$dpa -c HOAF /dev/stdout;"}, executed with \texttt{OpenJDK IcedTea 2.6.6, java version 1.7.0\_101}} \cite{DBLP:conf/cav/TsaiTH13} and the Streett-based construction StreetIAR of \cite{loding-thesis}.
As we are not aware of any implementations of StreetIAR, we implemented it ourselves in Haskell\footnote{\label{footnote:haskell_version}Compiled with GHC 7.10.3.}.
Both of these constructions are using state-based acceptance.
In order to allow for a fair comparison, we therefore also implemented sbIAR, a variant of our construction working directly with state-based acceptance\footnote{We also proved correctness for the direct construction, the proof can be obtained by trivial modifications of the proofs in this paper.} in Haskell\cref{footnote:haskell_version}, too.
Additionally, we combine every tool with Spot's multi-purpose post-processing\footnote{\texttt{autfilt -{}-deterministic -{}-generic -{}-small -{}-high}} and denote this by a subscript $P$ (for post-processing), \abbreveg IAR$^*$ combined with this post-processing is written IAR$^*_P$.

\smallskip

In \cref{tbl:experimental_dra_to_dpa_goal} we present a comparison between GOAL, StreettIAR and our unoptimized state-based implementation sbIAR. Additionally, since GOAL does not perform too well, we also include its post-processed variant GOAL$_P$.
For comparison, we also include our optimized variant $\IAR^*_P$. As test data, we use 1000 state-based DRA over 4 atomic propositions with 5 to 15 states, a transition density of $0.05$ and 2 to 3 Rabin pairs\footnote{\texttt{randaut 4 -{}-seed=0 -Q 5..15 -{}-acceptance="Rabin 2..3" -{}-density=0.05 -{}-deterministic -{}-acc-probability 0.2 -{}-state-based-acceptance -{}-hoaf -n1000}.
The acceptance probability parameter denotes the chance of a particular transition belonging to a Rabin pair.}.
We use Spot's tool \texttt{autfilt} to gather the statistics.
Failures denote either time-outs, out of memory errors or invalid results, \abbreveg automata which could not be read by \texttt{autfilt}, which sometimes occurred with GOAL.
\smallskip
\begin{table}[t]
	\caption{Comparison of the DRA to DPA translations on 1000 randomly generated DRAs.
		First, we compare the cases where all tools finished successfully, according to the average size, the number of SCCs and the run-time.
		Second, we give the percentage each tool produces an automaton with the least number of states, and failures, respectively.}
	\centering
	\begin{tabu} to 0.98\linewidth {llRRRRR}
		\toprule
		&          & GOAL & GOAL$_P$ & StreettIAR & sbIAR & $\IAR^*_P$ \\
		\midrule
		\multirow{3}{*}{\rotatebox[origin=c]{90}{avg}\enspace}
		& \#states & 1054 &      281 &       18.4 &  15.4 &       8.83 \\
		& \#SCC    & 73.2 &     19.2 &       4.97 &  4.33 &       1.61 \\
		& time {\scriptsize (s)}
		           & 11.7 &     15.7 &       0.02 &  0.02 &       0.99 \\
		\midrule
		\multicolumn{2}{l}{smallest {\scriptsize (\%)}}
		           & 15.5 &     37.8 &        7.7 &  15.5 &       95.9 \\
		\multicolumn{2}{l}{failure  {\scriptsize (\%)}}
		           &  8.6 &     11.9 &          0 &     0 &          0 \\
		\bottomrule
	\end{tabu}
	\label{tbl:experimental_dra_to_dpa_goal}
\end{table}

From the results in \cref{tbl:experimental_dra_to_dpa_goal} we observe that on this dataset all appearance-record variants drastically outperform GOAL.
We remark that IAR$^*$ performs even better compared to GOAL if more SCCs are involved.
However, for reasonably complex automata, virtually every execution of GOAL timed out or crashed, making more specific experiments difficult.
Already for the automata in \cref{tbl:experimental_dra_to_dpa_goal} with 5--15 states, GOAL regularly consumed around 3 GB of memory and needed roughly 10 seconds to complete on average, whereas our methods only used a few hundred MB and less than a second.
We could not find a dataset where GOAL showed a significant advantage over our new methods.
Therefore, we exclude GOAL from further experiments.
The remaining methods are investigated more thoroughly in the next experiment.

\smallskip
\begin{table}[t]
	\caption{Comparison of StreettIAR and (sb)IAR on 1000 randomly generated DRAs.
		We use the same definitions as in \cref{tbl:experimental_dra_to_dpa_goal}.}
	\centering
	\begin{tabu} to 0.98\linewidth {llRRRRR}
		\toprule
		&          &  StreettIAR & sbIAR & StreettIAR$^*_P$ & sbIAR$^*_P$ & $\IAR^*_P$ \\
		\midrule
		\multirow{3}{*}{\rotatebox[origin=c]{90}{avg}\enspace}
		& \#states &        4959 &  1568 &             4175 &        1081 &        833 \\
		& \#SCC    &        63.8 &  42.5 &             1.35 &        1.35 &       1.35 \\
		& time {\scriptsize (s)}
		           &        1.86 &  0.34 &            39.47 &        3.11 &       3.38 \\
		\midrule
		\multicolumn{2}{l}{smallest {\scriptsize (\%)}}
		           &           0 &     0 &              0.4 &        5.90 &       95.1 \\
		\multicolumn{2}{l}{failure  {\scriptsize (\%)}}
		           &         1.3 &     0 &              1.4 &           0 &          0 \\
		\bottomrule
	\end{tabu}
	\label{tbl:experimental_dra_to_dpa_sar}
\end{table}

In \cref{tbl:experimental_dra_to_dpa_sar} we compare StreettIAR to sbIAR on more complex input automata to demonstrate the advantages of our new method compared to the existing StreettIAR construction.
We consider the methods both in the basic setting and with post-processing and optimizations.
Note that as the presented optimizations are applicable to appearance records in general, we also added them to our implementation of StreettIAR.
Its optimized version is denoted by StreettIAR$^*$.
Again, we include our best (transition-based) variant $\IAR^*_P$ for reference.
The dataset now contains DRA with 20 to 30 states\footnote{\texttt{randaut 4 --seed=0 -Q 20..30 --acceptance="Rabin 6" --density=0.05 --acc-probability=0.2 --deterministic --state-based-acceptance --hoaf -n1000}}.

StreettIAR is significantly outperformed by our new methods in this experiment.
This is quite surprising, considering that both methods essentially follow the same idea of index appearance records, only from different perspectives.
The difference is partially due to \cref{rem:streett}.
Besides, we have observed that the discrepancy between StreettIAR and IAR is closely linked to the amount of acceptance pairs.
After increasing the number of pairs further, the gap between the two approaches grows dramatically.
For instance, on a dataset of automata with 8 states and 8 Rabin pairs, the IAR construction yielded automata roughly an order of magnitude smaller: sbIAR needed less than three hundred states compared to StreettIAR needing over three thousand.
Applying the post-processing does not remedy the situation.

\smallskip
\begin{table}[t]
	\caption{Evaluation of the presented optimizations on 1000 randomly generated DRAs, again using the same definitions as in \cref{tbl:experimental_dra_to_dpa_goal}.
	No tool failed for any of the input automata.}
	\centering
	\begin{tabu} to 0.98\linewidth {llRRRRRR}
		\toprule
		&          & sbIAR & sbIAR$^*_P$ & $\IAR$ & $\IAR_P$ & $\IAR^*$ & $\IAR^*_P$ \\
		\midrule
		\multirow{3}{*}{\rotatebox[origin=c]{90}{avg}\enspace}
		& \#states &  3431 &        2530 &   1668 &     1655 &     1302 &       1296 \\
		& \#SCC    &  24.8 &        1.14 &   8.98 &      3.5 &     1.43 &       1.43 \\
		& time {\scriptsize (s)}
		           &  0.77 &       11.47 &   1.09 &     48.3 &     76.5 &       95.7 \\
		           \midrule
		\multicolumn{2}{l}{smallest {\scriptsize (\%)}}
		           &     0 &           0 &   38.3 &    48.30 &     76.5 &       95.7 \\
		\bottomrule
	\end{tabu}
	\label{tbl:experimental_dra_to_dpa_opt}
\end{table}

Finally, we demonstrate the significance of the transition-based acceptance and our optimizations in \cref{tbl:experimental_dra_to_dpa_opt}.
To evaluate the impact of our improvements, we compare the unoptimized $\IAR$ procedure and its post-processed counterpart to the optimized $\IAR^*$ and $\IAR^*_P$.
Furthermore, we also include our state-based version in its basic (sbIAR) and best (sIAR$^*_P$\footnote{We use \texttt{autfilt -{}-state-based-acceptance} to convert the transition based input DRA to state based.}) form.
We run these algorithms on a dataset of DRA with 20 states each\footnote{\texttt{randaut 4 --seed=0 -Q 20 --acceptance="Rabin 5" --acc-probability=0.05 --density=0.1 --deterministic --hoaf -n1000}}.

Spot's generic post-processing algorithms often yield sizeable gains, but they are marginal compared to the effect of our optimizations on this dataset. Our optimizations are thus not only significantly beneficial, but also irreplacable by general purpose optimizations.
We furthermore want to highlight the reduction of SCCs.
As a final remark, we emphasize the improvements due to the adoption of transition-based acceptance, halving the size of the automata.

\subsection{Linear Temporal Logic}

Motivated by the previous results we concatenated $\IAR^*$ with Rabinizers LTL-to-DRA translation, obtaining an LTL-to-DPA translation.
We compare this approach to the established tool \texttt{ltl2tgba} of Spot, which can also produce DPA\footnote{By specifying \texttt{-{}-deterministic -{}-generic} on the command line}.
We use Spot's comparison tool \texttt{ltlcross} in order to produce the results.
Unfortunately, this tool sometimes crashes caused by too many acceptance sets\footnote{Around 20 acceptance sets. The exact error message emitted is\\\texttt{-terminate called after throwing an instance of 'std::runtime\_error'\\
~~what():  Too many acceptance sets used.}}.
We alleviated this problem by splitting our datasets into smaller chunks.
Time-outs are set to 15 minutes.

First, we compare the two tools on random LTL formulae.
We use \texttt{randltl} and \texttt{ltlfilt} to generate pure LTL formulae\footnote{\texttt{randltl -n2000 5 -{}-tree-size=20..25 --seed=0 -{}-simplify=3 -p -{}-ltl-priorities 'ap=3,false=1,true=1,not=1,F=1,G=1,X=1,equiv=1\\,implies=1,xor=0,R=0,U=1,W=0,M=0,and=1,or=1' | ltlfilt -{}-unabbreviate="eiMRW\^"}}.
The test results are outlined in \cref{tbl:rabinizer_vs_ltl2tgba_rand}.
On average, our methods are comparable to \texttt{ltl2tgba}, even outperforming it slightly in the number of states.

Note that the averages have to be compared carefully.
As the constructions used by \texttt{ltl2tgba} are fundamentally different from ours, there are some formulae where we outperform \texttt{ltl2tgba} by orders of magnitude and similarly in the other direction.
We conjecture that on some formulae \texttt{ltl2tgba} has an edge merely due to its rewriting together with numerous pre- and post-processing steps, whereas our method profits from Rabinizer, which can produce smaller deterministic automata also for very complex formulae.
On many dataset we tested, median state count over all formulae (including timeouts) is better for our methods.
For more detail, see the histogram in \cref{sec:app:data}, \cref{fig:app:graph}.

\begin{table}[t]
	\caption{Comparison of \texttt{ltl2tgba} to Rabinizer + $\IAR^*_P$ on 2000 LTL formulae.}
	\centering
	\begin{tabu} to 0.75\linewidth {llrp{5pt}r}
		\toprule
		&          & Rabinizer + $\IAR^*_P$ & & \texttt{ltl2tgba} \\
		\midrule
		\multirow{3}{*}{\rotatebox[origin=c]{90}{avg}\enspace}
		& \#states &                   6.60 & &            7.89 \\
		& \#acc    &                   2.31 & &            1.79 \\
		& \#SCC    &                   4.49 & &            4.69 \\
		& timeouts &                     22 & &               0 \\
	\end{tabu}
	\label{tbl:rabinizer_vs_ltl2tgba_rand}
\end{table}

To give more insight in the difference between the approaches, we list several classes of formulae where our technique performs particularly well.
For instance, for fairness-like constraints our toolchain produces significantly smaller automata than \texttt{ltl2tgba}, see \cref{tab:fair}.
Further examples, previously investigated in e.g. \cite{DBLP:conf/cav/KretinskyE12,DBLP:conf/atva/BabiakBKS13,DBLP:conf/cav/EsparzaK14} can be found in \cref{sec:app:data}, \cref{tab:app:fair}, including formulae of the GR(1) fragment \cite{DBLP:conf/vmcai/PitermanPS06}.
Additionally, our method is performing better on many practical formulae, for instance complex formulae from \textsc{Spec Pattern} \cite{DBLP:conf/icse/DwyerAC99}\footnote{Spec Patterns: Property Pattern Mappings for {LTL}. \url{http://patterns.projects.cis.ksu.edu/documentation/patterns/ltl.shtml}}.

\begin{table}[t]
	\caption{Fairness formulae: $\mathit{Fairness}(k)=\bigwedge_{i=1}^k (\G\F a_i\Rightarrow \G\F b_i)$}
	\label{tab:fair}
	\setlength{\tabcolsep}{4pt}
	\begin{tabu} to 0.9\textwidth {Lrrrcrrr}
		& \multicolumn{3}{c}{Rabinizer+$\IAR^*_P$} & & \multicolumn{3}{c}{\texttt{ltl2tgba}} \\
		Formula           & States & Acc. & SCCs & & States &  Acc. & SCCs \\
		\hline
		$\mathit{Fairness}(1)$
		& 2 & 4 & 1 & & 5 & 4 & 3 \\
		$\mathit{Fairness}(2)$
		& 12 & 9 & 1 & & 44 & 8 & 9 \\
		$\mathit{Fairness}(3)$
		& 1431 & 17 & 1 & & 8607 & 20 & 546 \\
		\hline
	\end{tabu}
\end{table}

%% file: appendix.tex
\appendix
\section{Proofs}

\renewcommand{\baselinestretch}{1}

Throughout the section, we write $\mathcal{R}$ to denote an arbitrary Rabin automaton $\tuple*{Q, \Sigma, \delta, q_0, \set*{\tuple*{F_i, I_i}}_{i=1}^k}$, $\mathcal{P} = \fun{\IAR}{\mathcal{R}} = \tuple{\tilde{Q}, \Sigma, \tilde{\delta}, \tilde{q}_0, \lambda}$ and $\mathcal{P}^* = \fun{\IAR^*}{\mathcal{R}} = \tuple*{Q^*, \Sigma, \delta^*, q_0^*, \lambda^*}$ the result of applying $\IAR$ and $\IAR^*$ to $\mathcal{R}$, respectively.
Furthermore, whenever we have a transition $\tilde{t} = \trans{\tuple{q, \pi}}{a}{\tuple{q', \pi'}}$ in the $\mathcal{P}$, we use $t = \trans{q}{a}{q'}$ to denote the corresponding transition in the original Rabin automaton (analogously for $t^*$ in $\mathcal{P}^*$).

\subsection{Proof of \cref{stm:iar_correctness,stm:iar_independence_of_initial_permutation}} \label{section:appendix:iar_correctness}

We prove that after constructing $\mathcal{P}$ from $\mathcal{R}$, the language of \emph{any} state $q \in Q$ (\abbrevie the language accepted by the automaton starting at $q$) is equal to the language of \emph{any} state $\tuple{q, \pi} \in \tilde{Q}$.
This trivially implies \cref{stm:iar_correctness,stm:iar_independence_of_initial_permutation}.
We will reuse this result in another proof.

Given an $\omega$-automaton $\Automaton = \tuple*{Q, \Sigma, \delta, q_0, \alpha}$ and $q \in Q$ we write $\Automaton_q = \tuple*{Q, \Sigma, \delta, q, \alpha}$ to denote the automaton with new initial state $q$.
\begin{lemma}\label{stm:iar_correctness_extended}
	We have that $\Language(\mathcal{P}_{\tuple*{q', \pi'}}) = \Language(\mathcal{R}_{q'})$ for all $q' \in Q$ and $\pi' \in \Pi^k$ such that $\tuple*{q', \pi'} \in \tilde{Q}$.
\end{lemma}
\begin{proof}
	Let $q' \in Q$ and $\pi' \in \Pi^k$ be as in the assumption and define $\tilde{q}' = \tuple{q', \pi'}$.
	We immediately see from the definition of $\IAR$ that for a word $w$, $\mathcal{R}_{q'}$ has a run $\rho$ on $w$ iff $\mathcal{P}_{\tilde{q}'}$ a run $\tilde{\rho}$ on it.
	Furthermore, for every $\rho_i = \trans*{q}{a}{q'}$ we have that $\tilde{\rho}_i = \trans{\tuple{q, \pi}}{a}{\tuple{q', \pi'}}$ for some $\pi, \pi' \in \Pi^k$.
	
	\smallskip
	\noindent \underline{$\Language(\mathcal{R}_{q'}) \subseteq \Language(\mathcal{P}_{\tilde{q}'})$}:
	Let $w \in \Language(\mathcal{R}_{q'})$ be a word accepted by the Rabin automaton $\mathcal{R}_{q'}$.
	Let $\rho$ and $\tilde{\rho}$ denote the runs of $\mathcal{R}_{q'}$ and $\mathcal{P}_{\tilde{q}'}$ on it, respectively.
	We show that any transition $\tilde{t} \in \fun*{\Inf}{\tilde{\rho}}$ with maximal priority (among all infinitely often visited transitions) has even priority and thus $w$ is also accepted by $\mathcal{P}_{\tilde{q}'}$.
	
	By assumption, there exists an accepting Rabin pair $\tuple{F_n, I_n}$, \abbrevie $I_n \in \fun*{\RequiredInf}{w}$, $F_n \notin \fun*{\ProhibitedInf}{w}$, and an infinitely often visited transition $\tilde{t}_n = \trans{\tuple{q_n, \pi_n}}{a}{\tuple{q_n', \pi_n'}}$ with $t_n \in (\fun*{\Inf}{\rho} \intersection I_n) \setminus F_n$.
	Hence, $\fun*{\pi_n^{-1}}{n} \leq \fun*{\maxInd}{\tilde{t}_n}$ by definition of $\maxInd$.
	
	Now, fix an arbitrary $\tilde{t} = \trans{\tuple{q, \pi}}{a}{\tuple{q', \pi'}} \in \fun*{\Inf}{\tilde{\rho}}$ with maximal $\fun*{\maxInd}{\tilde{t}}$ among all the infinitely often visited transitions, \abbrevie $\fun*{\maxInd}{\tilde{t}_n} \leq \fun*{\maxInd}{\tilde{t}}$.
	From \cref{stm:iar_inf_visited_indices} we know that the position of pair $n$ stays constant along the infinitely often visited states, \abbrevie $\fun*{\pi^{-1}}{n} = \fun*{\pi_n^{-1}}{n}$.
	Together, this yields $\fun*{\pi^{-1}}{n} = \fun*{\pi_n^{-1}}{n} \leq \fun*{\maxInd}{\tilde{t}_n} \leq \fun*{\maxInd}{\tilde{t}}$.
	
	Assume for contradiction that $\fun*{\lambda}{\tilde{t}}$ is odd, \abbrevie $\tilde{t} \in F_{\fun*{\pi}{\fun*{\maxInd}{\tilde{t}}}}$.
	By \cref{stm:smaller_indices_visited_inf_often} this yields $\set*{F_{\fun*{\pi}{i}} \mid i \leq \fun*{\maxInd}{\tilde{t}}} \subseteq \ProhibitedInf$.
	As we previously argued $\fun*{\pi^{-1}}{n} \leq \fun*{\maxInd}{\tilde{q}}$ and therefore $F_n \in \ProhibitedInf$, contradicting the assumption.
	
	\smallskip
	\noindent \underline{$\Language(\mathcal{P}_{\tilde{q}'}) \subseteq \Language(\mathcal{R}_{q'})$}:
	Let $w \in \Language(\mathcal{P}_{\tilde{q}'})$ be a word accepted by the constructed parity automaton.
	Again, denote the corresponding runs by $\rho$ and $\tilde{\rho}$.
	We show that there exists some $n$ where $F_n \notin \ProhibitedInf$ and $I_n \in \RequiredInf$, \abbrevie $\mathcal{R}_{q'}$ accepts $w$.
	
	By assumption, the maximal priority $\lambda$ of all infinitely often visited states is even.
	Let $\tilde{t} \in \fun*{\Inf}{\tilde{\rho}}$ be a state with $\fun*{\lambda}{\tilde{t}} = \lambda = 2 \cdot \fun*{\maxInd}{\tilde{t}}$, \abbrevie $\tilde{t}$ is a witness for the maximal priority.
	Defining $n = \fun*{\pi}{\fun*{\maxInd}{\tilde{t}}}$ we have that $t \in I_n \setminus F_n$ by definition of the priority assignment.
	Clearly, $t \in \fun*{\Inf}{\rho}$ and hence $I_n$ is visited infinitely often in the Rabin automaton (via $t$).
	We now show that $F_n$ is visited only finitely often.
	
	Assume the contrary, \abbrevie that $F_n$ is visited infinitely often.
	This implies that every time after taking $\tilde{t}$, some transition $\tilde{t}_F = \trans{\tuple{q_F, \pi_F}}{a}{\tuple{q_F', \pi_F'}}$ with $t_F \in F_n$ is eventually taken.
	After visiting $I_n$, the position of $F_n$ can not decrease until it is visited for the first time by definition of $\tilde{\delta}$.
	Hence for each visit of $\tilde{t}$ we can choose a $\tilde{t}_F$ such that $\fun*{\pi_F^{-1}}{n} \geq \fun*{\maxInd}{\tilde{t}}$.
	But then also $\fun*{\maxInd}{\tilde{t}_F} \geq \fun*{\maxInd}{\tilde{t}}$, as $q_F \in F_n$.
	Hence $\fun*{\lambda}{\tilde{q}_F} > \fun*{\lambda}{\tilde{q}}$, contradicting the assumption of $\fun*{\lambda}{\tilde{q}}$ being maximal.\qed
\end{proof}

\subsection{Proof of Complexity}\label{app:complexity}

\begin{proof}
	It has been shown \cite[Theorem 7]{DBLP:conf/fsttcs/Loding99} that there exists a family $\sequence{\Language_n}_{n \geq 2}$ of languages such that for every $n$ the language $\Language_n$ can be recognized by a DSA with $O(n)$ states and $O(n)$ pairs, but cannot be recognized by a DRA with less than $n! $ states.
	By interpreting the DSA as DRA for $\overline{\Language_n} = \Sigma^{\omega} \setminus \Language_n$, the statement also holds when transforming DRA to DPA.
	Let $\mathcal{R}$ be the Streett automaton interpreted as DRA, \abbrevie it accepts $\overline{\Language_n}$ with $O(n)$ states and $O(n)$ Rabin pairs.
	Assume for contradiction that a DPA $\mathcal{P}$ recognizes $\overline{\Language_n}$ with less than $n! $ states.
	One can easily verify that a DPA is complemented by using the priority function $\fun{\lambda'}{t} = \fun{\lambda}{t} + 1$.
	Applying this to $\mathcal{P}$ yields a DPA and thus a DRA recognizing $\Language_n$ with less than $n!$ states, a contradiction.\qed
\end{proof}

\subsection{Proof of \cref{stm:iar_star_correctness}} \label{section:appendix:iar_star_correctness}

In order to prove correctness of $\IAR^*$, we prove some auxiliary lemmas.
We use $\Pi_k = \Union_{i=0}^k \Pi^i$ to denote all permutations of length up to $k$ including the empty permutation $\varepsilon$.
With this, $Q^* \subseteq Q \times \Pi$.

First and foremost, we immediately see that the algorithm is well defined in the sense that $\delta^*$ only gets assigned at most one value for each pair $\tuple*{q^*, a}$ and $\lambda^*$ gets assigned exactly one value for each $q^* \in Q^*$.
Now we show that $\IAR^*$ emulates runs on the original automaton, \abbrevie every run of a Rabin automaton has a unique corresponding run in its $\IAR^*$ translation.
\begin{lemma}\label{stm:iar_star_simulates_rabin_run}
	$\mathcal{R}$ has a run $\rho$ on a word $w$ iff $\mathcal{P}^*$ has a run $\rho^*$ on it.
	For a given $w$, this $\rho^*$ is unique (if it exists) and for every $\rho^*_i = \trans{\tuple{q, \pi}}{a}{\tuple{q', \pi'}}$ we have that $\rho_i = \trans{q}{a}{q'}$ for some $\pi, \pi' \in \Pi$.
\end{lemma}
\begin{proof}
	We first establish existence of corresponding runs.
	
	\smallskip
	\noindent \underline{$\Rightarrow$}: We show that for all $q, q' \in Q$ and $a \in \Sigma$ with $q' = \fun{\delta}{q, a}$ and for every $\pi \in \Pi$ such that $\tuple{q, \pi}$ is reachable from $q^*_0$, there is a $\pi' \in \Pi$ with $\tuple*{q', \pi'} = \fun{\delta^*}{\tuple*{q, \pi}, a}$, \abbrevie the run of $\mathcal{P}^*$ cannot get \enquote{stuck}.
	
	Assume for contradiction that there are such $q$, $\pi$, $a$ and $q'$ where $\fun{\delta^*}{\tuple*{q, \pi}, a}$ is undefined.
	As by assumption $\tuple*{q, \pi}$ is reachable in the resulting automaton, it has been added to $Q^*$ while processing a particular SCC $S$ of the original automaton.
	We distinguish multiple cases:
	\begin{itemize}
		\item $q' \notin S$: While processing all SCCs, the algorithm eventually adds $\tuple*{q', \pi'}$ with some $\pi'$ to $Q^*$.
		As $q$ and $q'$ belong to different SCCs and $q' = \fun{\delta}{q, a}$ by assumption, \cref{alg:iar_star:line:update_delta_star_with_inter_scc_connection} is visited with these particular values and $\fun{\delta^*}{\tuple*{q, \pi}, a} = \tuple*{q', \pi'}$.
		This yields the contradiction.
		\item $q' \in S$: By definition we have that $q' = \fun{(\delta \upharpoonright S)}{q, a}$.
		Again, we do a case distinction:
		\begin{itemize}
			\item $\set*{i \mid I_i \intersection (\Delta \upharpoonright S) = \emptyset} = \emptyset$: Only $\tuple*{q, \varepsilon}$ is added to $Q^*$ (in \cref{alg:iar_star:line:update_q_star_in_trivial_scc}) and the corresponding transition is added in \cref{alg:iar_star:line:update_delta_star_in_trivial_scc}.
			\item Otherwise: The algorithm invokes $\IAR$ on the sub-automaton (containing the transition $\trans{q}{a}{q'}$). The definition of $\IAR$ immediately gives a contradiction.
		\end{itemize}
	\end{itemize}
	By a simple inductive argument, existence of the run on a particular word $w$ follows.
	
	\smallskip
	\noindent \underline{$\Leftarrow$}: By investigating the algorithm, one immediately sees that if $\fun{\delta^*}{\tuple*{q, \pi}, a}$ is assigned some value $\tuple*{q', \pi'}$, $q' = \fun{\delta}{q, a}$ is a precondition to that.
	
	\medskip
	
	We now prove uniqueness of the run.
	Assume for contradiction that for some $\tuple*{q, \pi}$ and letter $a$ the algorithm added both $\tuple*{q', \pi'}$ and $\tuple*{q'', \pi''}$ as successors with $\pi' \neq \pi''$ or $q' \neq q''$.
	Using the same reasoning as in the \enquote{$\Leftarrow$} proof, we see that $q' = q''$.
	Let now $S$ be the SCC in the original DRA with $q \in S$.
	Again, we use a case distinction:
	\begin{itemize}
		\item $q' \in S$: In the \enquote{if}-branch only one successor is added by the algorithm.
		In the \enquote{else}-branch, the statement follows from the definition of $\IAR$.
		\item $q' \notin S$: The for loop connecting the SCCs together is entered exactly once with the variables $q$, $q'$ and $a$.
		The algorithm picks any $\pi'$ as the permutation of the successor, but only this single one.\qed
	\end{itemize}
\end{proof}
From this we get as immediate consequences that $\IAR^*$ indeed outputs a deterministic parity automaton and yields a total automaton if $\mathcal{R}$ was total.
Furthermore, we show that every SCC in the result corresponds to a subset of a SCC in the original automaton.
In other words, it cannot be the case that a SCC in the resulting automaton contains states corresponding to states in the Rabin automaton in two different SCCs.
\begin{corollary} \label{stm:iar_star_scc_correspondence}
	For any SCC $S^* \subseteq Q^*$ in $\mathcal{P}^*$, we have that its projection $\set*{q \in Q \mid \Exists \pi \in \Pi_k. \tuple*{q, \pi} \in S^*}$ to $\mathcal{R}$ is a subset of some SCC in $\mathcal{R}$.
\end{corollary}
\begin{proof}
	Consider an arbitrary cycle in $\mathcal{P}^*$.
	Projecting the cycle to $\mathcal{R}$ again results in a cycle by \cref{stm:iar_star_simulates_rabin_run}.
\end{proof}
As a last lemma, we prove that $\mathcal{R} \upharpoonright_q S$ behaves as it should.

\begin{lemma} \label{stm:rabin_restriction_correct}
	Let $w$ be a word such that $\mathcal{R}$ has a run $\rho$ on it.
	Let $S$ be the SCC containing $\fun{\Inf}{w}$ and pick an arbitrary $q \in S$.
	Fix $j \in \Naturals$ such that $\rho_i \in \Delta \upharpoonright S$ for all $i \geq j$.
	Then $w$ is accepted by $\mathcal{R}$ iff $w' = w_j w_{j+1} \dots$ is accepted by $(\mathcal{R} \upharpoonright_q S)_{\rho_j}$.
\end{lemma}
\begin{proof}
	Fix $w$, $\rho$, $S$, $q$, $j$ and $w'$ as stated.
	One immediately sees that $(\mathcal{R} \upharpoonright_{q} S)_{\rho_j}$ has a run $\rho' = \rho_j \rho_{j+1} \dots$ on $w'$ and $\fun{\Inf}{\rho} = \fun{\Inf}{\rho'}$.
	Hence we only need to show that there are pairs in both automata accepting the respective runs.
	
	\smallskip
	\noindent \underline{$\Rightarrow$}: As $w$ is accepted by $\mathcal{R}$ there is an accepting Rabin pair $\tuple*{F_i, I_i}$.
	By assumption, $\fun{\Inf}{\rho} \subseteq \Delta \upharpoonright S$ and $\fun{\Inf}{\rho} \intersection I_i \neq \emptyset$.
	Hence, $I_i \intersection (\Delta \upharpoonright S) \neq \emptyset$ and $\tuple*{F_i \intersection (\Delta \upharpoonright S), I_i \intersection (\Delta \upharpoonright S)}$ is a pair of the restricted automaton accepting $\rho'$.
	
	\smallskip
	\noindent \underline{$\Leftarrow$}: Trivial.\qed
\end{proof}
With these results, we prove the correctness of the algorithm.
\begin{proof}[of \cref{stm:iar_star_correctness}]
	Let $w \in \Sigma^\omega$ be an arbitrary word.
	By \cref{stm:iar_star_simulates_rabin_run}, we have that $\mathcal{R}$ has a run $\rho$ on $w$ iff $\mathcal{P}^*$ has a run $\rho^*$ on it.
	Assume \abbrevwlog that both automata indeed have such runs (otherwise $w$ trivially is not accepted by neither automata).
	Let $S$ and $S^*$ be the SCCs containing $\fun{\Inf}{\rho}$ and $\fun{\Inf}{\rho^*}$, respectively.
	We further assume \abbrevwlog that $\set*{i \mid I_i \intersection (\Delta \upharpoonright S) \neq \emptyset} \neq \emptyset$, otherwise $w$ trivially is not accepted by neither of the automata, as both of them only generate uneven priorities infinitely often.
	
	By virtue of \cref{stm:iar_star_scc_correspondence}, the SCC $S^*$ is constructed while processing $S$ in the main loop, \abbrevie it is a SCC of $\mathcal{P}' = \fun{\IAR}{\mathcal{R}', \fun{\pickPerm}{q, S}}$ where $\mathcal{R}' = \mathcal{R} \upharpoonright_q S$.
	As both runs eventually remain in the respective SCCs, there is a $j \in \Naturals$ such that $\rho_i \in \Delta \upharpoonright S$ and $\rho^*_i \in \Delta^* \upharpoonright S$ for all $i \geq j$.
	By \cref{stm:rabin_restriction_correct} we have that $w' = w_j w_{j+1} \dots$ is accepted by $\mathcal{R}'$ iff $w$ is accepted by $\mathcal{R}$.
	Furthermore, employing \cref{stm:iar_correctness_extended} we have $w'$ is accepted by $\mathcal{P}'_{\rho^*_j}$ iff it is accepted by $\mathcal{R}'$.
	By construction, $w$ is accepted by $\mathcal{P}^*$ iff $w'$ is accepted by $\mathcal{P}'_{\rho^*_j}$.
	Together, this yields that $w$ is accepted by $\mathcal{R}$ iff it is accepted by $\mathcal{P}^*$. \qed
\end{proof}
\subsection{Proof of \cref{stm:iar_star_subautomaton_bscc_unique}} \label{section:appendix:bscc_unique}
For increased readability, we extend the notion of runs to finite words: We write $\Sigma^*$ to denote the set of all finite words over a given alphabet $\Sigma$.
The length of such a word $w$ is denoted by $\length{w}$.
We say that an automaton $\Automaton = \tuple*{Q, \Sigma, \delta, q_0, \alpha}$ has the run $\rho$ on a finite word $w \in \Sigma^*$ starting from $q$ if $\rho = \tuple*{\rho_0, \rho_1, \dots, \rho_{\length{w}}}$ where $\rho_0$ starts at $q_0$, $\rho_i$ moves under $w_i$ for all $0 \leq i < \length{w}$ and $\rho_{i}$ starts at the same state as $\rho_{i-1}$ ends for all $0 < i < \length{w}$.
\begin{proof}
	If either $\cardinality{Q} = 1$ or $k \leq 1$, the graph of the constructed automaton is isomorphic to the graph of the input automaton.
	Hence assume that $\cardinality{Q} > 1$ and $k > 1$.
	For now, fix an arbitrary initial permutation $\pi_0$.
	Note that the relative ordering of all positions of pairs with empty prohibited sets stays the same for all runs as they are never visited.
	
	\medskip
	
	We first show existence of BSCCs.
	As $Q$ is assumed to be strongly connected, every state in $\mathcal{R}$ necessarily has a successor.
	From the definition of $\IAR$ it immediately follows that firstly $\mathcal{P}$ contains at least one SCC and secondly that there can be no state without a successor, which implies that there always exists a BSCC.
	
	\medskip
	
	Next, we show both uniqueness and the given characterization of BSCCs.
	Let $S$ be a BSCC of $\mathcal{P}$, $S'$ a SCC with $Q = \set*{q \mid \Exists \pi \in \Pi^k. \tuple*{q, \pi} \in S'}$ and $S \neq S'$. This implies that $S \intersection S' = \emptyset$, as different SCCs are disjoint by definition.
	Note that $S'$ may also be another BSCC.
	
	As we assumed that $\cardinality{Q} > 1$ we also have $\cardinality{S'} > 1$.
	Intuitively, this means that the graph corresponding to $S'$ contains a path which visits all states in $S'$.
	Formally, for each state ${q^*}' \in S'$ we can find a finite word $w$ such that the run of $\mathcal{P}$ on $w$ starting from ${q^*}' \in S'$ visits each state in $S$ at least once and ends in ${q^*}'$.
	
	Fix any ${q^*}' = \tuple*{q, \pi'} \in S'$ and let $w$ be such a word.
	From the definition of $\IAR$, we can conclude that each BSCC contains all states of the original automaton (as we assumed that $Q$ is a SCC).
	Hence, choose some $\pi \in \Pi^k$ such that $q^* = \tuple*{q, \pi} \in S$.
	After following the word $w$ on $\mathcal{P}$ starting from the two states $q^*$ and ${q^*}'$, we arrive at $\tuple*{q, \overline{\pi}} \in S$ and $\tuple*{q, \overline{\pi}'} \in S'$, respectively.
	But, as every state and thus every non-empty $F_i$ was visited along the path, we have that $\overline{\pi} = \overline{\pi}'$.
	Therefore $\tuple*{q, \overline{\pi}} = \tuple*{q, \overline{\pi}'}$ and hence $S \intersection S' \neq \emptyset$, contradicting the assumption.
	
	\medskip
	
	Finally, we show that the BSCC is essentially unique independently of $\pi_0$: Choose $\pi_0 \in \Pi^k$ arbitrary and let $S$ be the BSCCs of $\mathcal{P} = \fun{\IAR}{\mathcal{R}, \pi_0}$.
	Clearly, for each $\tuple*{q, \pi} \in S$ the positions of all empty $F_i$ in $\pi$ is bigger than any of the visited ones, \abbrevie $\fun{\pi^{-1}}{j} < \fun{\pi^{-1}}{i}$ for all $i, j$ \abbrevst $F_j \neq \emptyset$, $F_i = \emptyset$.
	Assume for contradiction that there is some state $\tuple{q, \pi} \in S$ where $\fun{\pi^{-1}}{j} > \fun{\pi^{-1}}{i}$.
	As $Q$ is a SCC and $F_j$ is not empty by assumption, there exists a run which visits $\tuple{q, \pi}$ and then visits some $\tuple{q', \pi'}$ with $q' \in F_j$.
	By definition, $j$ is moved to the front and can never be moved to behind $i$.
	Hence, after visiting $\tuple{q', \pi'}$ no run can visit $\tuple{q, \pi}$ again, contradicting the assumption that $\tuple{q, \pi}$ and $\tuple{q', \pi'}$ are contained in the same SCC.
	
	Let now $\pi_0' \in \Pi^k$ arbitrary and let $S'$ be the BSCCs of $\mathcal{P}' = \fun{\IAR}{\mathcal{R}, \pi_0'}$.
	We prove by contradiction that $S$ and $S'$ are the same up to their relative ordering of the empty $F_i$ (which trivially always stays the same).
	Assume \abbrevwlog that there is some $q^* = \tuple*{q, \pi} \in S$ such that the ${q^*}' = \tuple*{q, \pi'}$ obtained by reordering the empty prohibited sets in $\pi$ is not in $S'$.
	We can pick a finite word $w$ such that the run of $\mathcal{P}$ on $w$ visits $q^*$, then visits all states in $S$ and finally ends in $q^*$.
	The run of $\mathcal{P}'$ on $w$ ends in some $\tuple*{q, \pi''}$.
	By now, after visiting $q^*$ every non-empty $F_i$ has been visited by the run of $\mathcal{P}'$.
	Additionally, we have shown that the empty prohibited sets are positioned at the end of every $\pi$ in the BSCC.
	Hence the order of all the non-empty $F_i$ in $\pi''$ is determined by the run after its visit of $q^*$.
	In other words, we have that $\pi'' = \pi'$, yielding a contradiction.\qed
\end{proof}
\vfill
\pagebreak
\section{More testing data}\label{sec:app:data}

\renewcommand{\G}{\textbf{G}}
\renewcommand{\F}{\textbf{F}}
\renewcommand{\X}{\textbf{X}}
\renewcommand{\U}{\textbf{U}}

\vspace{-1em}

\begin{table}[h]
	\centering
	\caption{Fairness and GR(1) formulae: \\
	$\mathit{Fairness}(k)=\bigwedge_{i=1}^k (\G\F a_i\Rightarrow \G\F b_i)$, $\mathit{Chained}(k)=\bigwedge_{i=1}^k (\G\F a_i\Rightarrow \G\F a_{i+1})$, $GR(1)_k=\big(\bigwedge_{i=1}^k \G\F a_i\big)\Rightarrow \big(\bigwedge_{i=1}^k \G\F b_i\big)$. \texttt{ltl2tgba} timed out on the last formula.}
	\label{tab:app:fair}
	\setlength{\tabcolsep}{4pt}
	\begin{tabu} to 0.98\textwidth {Lrrrp{12pt}rrr}
		\toprule
		& \multicolumn{3}{c}{Rab+$\IAR^*_P$} & & \multicolumn{3}{c}{\texttt{ltl2tgba}} \\
		Formula           & States & Acc. & SCCs & & States & Acc. & SCCs \\
		\midrule
		$\mathit{Fairness}(1)$
		& 2 & 4 & 1 & & 5 & 4 & 3 \\
		$\mathit{Fairness}(2)$
		& 12 & 9 & 1 & & 44 & 8 & 9 \\
		$\mathit{Fairness}(3)$
		& 1431 & 17 & 1 & & 8607 & 20 & 546 \\
		\midrule
		$\mathit{Chained}(2)$
		& 6 & 6 & 1 & & 15 & 6 & 3 \\
		$\mathit{Chained}(3)$
		& 78 & 11 & 1 & & 128 & 8 & 9 \\
		\midrule
		$\mathit{GR(1)}_1$ & 2 & 4 & 1 & & 5 & 4 & 3 \\
		$\mathit{GR(1)}_2$ & 6 & 4 & 1 & & 13 & 4 & 5 \\
		$\mathit{GR(1)}_3$ & 24 & 4 & 1 & & 52 & 4 & 9 \\
		$\mathit{GR(1)}_4$ & 120 & 4 & 1 & & 265 & 4 & 17 \\
		$\mathit{GR(1)}_5$ & 720 & 4 & 1 & & 1636 & 4 & 33 \\
		$\mathit{GR(1)}_6$ & 5040 & 4 & 1 & & - & - & - \\
		\bottomrule
	\end{tabu}
\end{table}%
\begin{table}[h]
	\vspace{-1em}
	\centering
	\caption{Spec patterns: \enquote{after Q until R} properties}
	\label{tab:app:sp}
	\setlength{\tabcolsep}{2pt}
	\begin{tabu} to 0.98\textwidth {Lrrrp{12pt}rrr}
		\toprule
		& \multicolumn{3}{c}{Rab+$\IAR^*_P$} & & \multicolumn{3}{c}{\texttt{ltl2tgba}} \\
		Formula           & States & Acc. & SCCs & & States & Acc. & SCCs \\
		\midrule
		\multicolumn{8}{p{0.98\textwidth}}{$\G(\neg q \vee ((r \vee \neg s \vee (\X((\G(r \vee \neg t)) \vee (\neg r \U (r \wedge (r \vee \neg t)))))) \U (p \vee r)) \vee (\G(\neg s \vee (\X(\G\neg t)))))$}\\
			& 4 & 2 & 1 & & 27 & 6 & 1 \\[0.25em]
		\multicolumn{8}{p{0.98\textwidth}}{$\G(\neg q \vee (\G p) \vee (\neg p \U (r \vee (\neg p \wedge s \wedge (\X(\neg p \U t))))))$}\\
			& 9 & 3 & 4 & & 19 & 6 & 7 \\[0.25em]
		\multicolumn{8}{p{0.98\textwidth}}{$\G(\neg q \vee ((\neg p \vee (\neg r \U (\neg r \wedge s \wedge (\X(\neg r \U t))))) \U (r \vee (\G(\neg p \vee (s \wedge (\X(\F t))))))))$}\\
			& 13 & 6 & 3 & & 11 & 6 & 2 \\[0.25em]
		\multicolumn{8}{p{0.98\textwidth}}{$\G(\neg q \vee ((\neg s \vee (\X((\G\neg t) \vee (\neg r \U (r \wedge \neg t)))) \vee (\X(\neg r \U (r \wedge (\F p))))) \U (r \vee (\G(\neg s \vee (\X((\G\neg t) \vee (\neg r \U (r \wedge \neg t)))) \vee (\X(\neg r \U (t \wedge (\F p)))))))))$}\\
			& 79 & 10 & 2 & & 282 & 8 & 4 \\
		\bottomrule
	\end{tabu}
	\vspace{-2em}
\end{table}

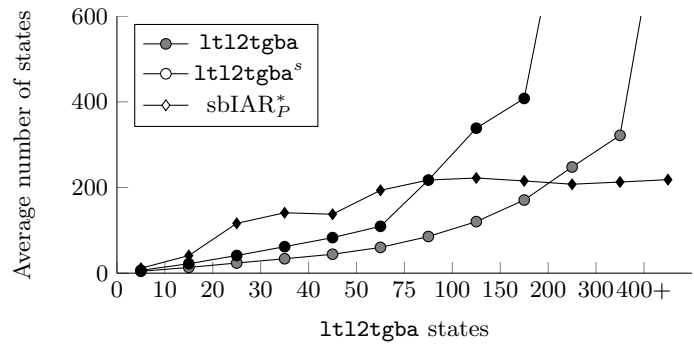
\begin{figure}[t]
	\vspace*{-2em}
	\centering
	\begin{tikzpicture}
		\begin{axis}[
				xmin=0, xmax=12,
				ymin=0, ymax=600,
				axis y line*=left,
				axis x line*=bottom,
				width=0.75\textwidth,
				height=5cm,
				xticklabels={0,10,20,30,40,50,75,100,150,200,300,400+},
				xtick={0,...,11},
				ylabel=Average number of states,
				xlabel=\texttt{ltl2tgba} states,
				legend pos=north west,
				legend style={draw=black,
				fill=white,align=left},
				cycle list name=black white]
				\addplot+[mark=*] coordinates {
			(0.5,   4.35)
			(1.5,  13.39)
			(2.5,  23.73)
			(3.5,  33.70)
			(4.5,  44.18)
			(5.5,  60.16)
			(6.5,  85.62)
			(7.5, 120.37)
			(8.5, 170.80)
			(9.5, 247.92)
			(10.5, 321.83)
			(11.5, 959.00)
			};
			\addplot+[mark=*,mark options=solid] coordinates {
			(0.5,   6.76)
			(1.5,  22.00)
			(2.5,  41.28)
			(3.5,  61.86)
			(4.5,  83.02)
			(5.5, 109.37)
			(6.5, 217.41)
			(7.5, 338.54)
			(8.5, 408.00)
			(9.5, 975.62)
			(10.5, 2070.83)
			(11.5, 2854.84)
			};
			\addplot+[mark=diamond*,mark options=solid] coordinates {
			(0.5,  11.47)
			(1.5,  41.13)
			(2.5, 116.34)
			(3.5, 141.05)
			(4.5, 137.66)
			(5.5, 193.51)
			(6.5, 217.51)
			(7.5, 222.27)
			(8.5, 215.40)
			(9.5, 207.69)
			(10.5, 212.83)
			(11.5, 218.37)
			};
			\legend{\texttt{ltl2tgba}, \texttt{ltl2tgba}$^s$, sbIAR$^*_P$}
		\end{axis}
	\end{tikzpicture}
	\caption{Comparing \emph{state based} IAR to \texttt{ltl2tgba} on a dataset of 5000 formulae, grouped by the amount of states \texttt{ltl2tgba} constructs.
	\texttt{ltl2tgba}$^s$ denotes the conversion of \texttt{ltl2tgba}'s output to state based acceptance.
	The time out was set to five minutes.
	Our observation was that prolonging the time out resulted in finishing more benchmarks, which were not significantly larger than the previously obtained average. Therefore, benchmarks where one of the two tools timed out are simply not considered here.
	Further, note that the x-axis is not linear.
	The missing data points of \texttt{ltl2tgba}$^s$ are 976, 2071 and 2855, for \texttt{ltl2tgba} it's 959.}
	\label{fig:app:graph}
\end{figure}